\documentclass[10pt,a4paper,twocolumn]{article}
\pdfoutput=1
\usepackage[left=1.76cm, right=1.76cm, top=2.50cm, bottom=2.50cm]{geometry}
\usepackage{titlesec}

\usepackage[T1]{fontenc}
\usepackage[utf8]{inputenc}
\usepackage{amsmath,amssymb,amsthm,mathtools}
\usepackage[libertine]{newtxmath}
\usepackage{microtype}
\usepackage{graphicx}

\theoremstyle{plain}
\newtheorem{theorem}{Theorem}

\newtheorem{proposition}{Proposition}

\theoremstyle{definition}
\newtheorem{definition}{Definition}
\newtheorem{assumption}{Assumption}
\newtheorem{remark}{Remark}

\DeclareMathOperator*{\esssup}{ess\,sup}

\begin{document}

	\title{\textbf{Robust stability of moving horizon estimation for continuous-time systems}}
	\author{Julian D. Schiller \and Matthias A. Müller}
	\date{%
		Leibniz University Hannover, Institute of Automatic Control\\%
		\textit{E-mail: \{schiller,mueller\}@irt.uni-hannover.de} \\[2ex]
	}

\maketitle

  \abstract{
	We consider a moving horizon estimation (MHE) scheme involving a discounted least squares objective for general nonlinear continuous-time systems.
	Provided that the system is detectable (incrementally integral input/output-to-state stable, i-iIOSS), we show that there exists a sufficiently long estimation horizon that guarantees robust global exponential stability of the estimation error in a {time-discounted} $L^2$-to-$L^\infty$ sense.
	In addition, we show that i-iIOSS Lyapunov functions can be efficiently constructed by verifying certain linear matrix inequality conditions.
	In combination, we propose a flexible Lyapunov-based MHE framework in continuous time, which particularly offers more tuning possibilities than its discrete-time analog, and provide sufficient conditions for stability that can be easily verified in practice.
	Our results are illustrated by a numerical example.
}

\section{Introduction} 

The internal state of a dynamic system is often not measured in practice for various (e.g., economic or physical) reasons.
However, control, system monitoring, and fault detection often require knowledge of the state, which in turn demands a reliable method for its reconstruction based on the available input and output signals.
This is generally challenging, especially when highly nonlinear systems are present and robustness to model errors and measurement noise must be ensured.
Moving horizon estimation---an optimization-based approach that inherently accounts for nonlinear system dynamics and constraints---is particularly suitable for this purpose.
Here, the estimate of the unknown state is produced by solving an optimization problem over past input-output data collected over a fixed finite time interval.
It can be interpreted as an approximation to full information estimation (FIE), which optimizes over all available historical data. However, the latter is usually only of theoretical interest (particularly as a benchmark for MHE), since the complexity of the underlying optimization problem continuously grows with time.

An MHE scheme for continuous-time systems was proposed in \cite{Michalska1995}. Since a cost function without a prior weighting was used (which can be seen as a regularization term), the system must satisfy an observability condition to ensure exponential convergence of the estimation error. 
Using such a cost function, however, requires long estimation horizons to ensure satisfactory performance in practice, cf.~\cite[Sec.~4.3.1]{Rawlings2017}.
Since the application of MHE inevitably requires some sort of sampling strategy (i.e., discrete time points at which the optimization is performed), schemes for discrete-time systems have recently been the main focus in the literature.
Early results have used certain nonlinear observability conditions and/or considered special cases such as no or asymptotically vanishing disturbances \cite{Rao2003,Alessandri2008,Rawlings2012}.
In recent years, the notion of incremental input/output-to-state stability (i-IOSS) has proven to be a very useful concept for nonlinear detectability, enabling significant advances in the MHE theory.
In particular, robust stability of MHE with a standard least squares objective was established in \cite{Mueller2017} and generalized in \cite{Allan2019a}, however, with theoretical guarantees that---counter-intuitively---deteriorate with an increasing estimation horizon.
The Lyapunov-based approach proposed in \cite{Allan2021a} is able to avoid this drawback, but on the other hand requires additional assumptions such as stabilizability.
In contrast, another line of research considers a cost function that includes explicit time discounting.
This establishes a more direct link to the i-IOSS property and allows the derivation of strong guarantees, cf., e.g., \cite{Knuefer2018,Knuefer2023,Hu2023}.
The Lyapunov framework proposed in~\cite{Schiller2023c}, which essentially relies on the same underlying principles, additionally provides less conservative conditions on the horizon length sufficient for guaranteed robustly stable estimation, cf. ~\cite[Sec. ~III.D]{Schiller2023c} for a more detailed discussion on this topic.

In this paper, we propose an MHE scheme using a discounted least squares objective for general nonlinear continuous-time systems and establish robust global exponential stability of the estimation error in a {time-discounted} $L^2$-to-$L^\infty$ sense.
Adapting the recent results for discrete-time settings, we use the concept of incremental integral input/output-to-state stability (i-iIOSS) introduced in~\cite{Schiller2023a} to characterize the underlying detectability property.
We additionally show how i-iIOSS Lyapunov functions can be efficiently computed by verifying certain linear matrix inequality (LMI)~conditions.

Especially when the physical system to be estimated actually corresponds to a continuous-time system (which is often the case in practice), the proposed MHE framework offers several key advantages over purely discrete-time schemes.
First, we note that arbitrary sampling strategies can be employed to define time instances at which the underlying optimization problem is actually solved, which can even be modified online at runtime.
{This provides a huge additional degree of freedom, and even allows the proposed MHE scheme to be used in an event-triggered fashion, that is, by choosing the sampling instances online depending on a suitable triggering rule.
	Consequently, the proposed MHE scheme can be better tailored to the problem at hand, which can yield more accurate results with less computational effort compared to standard equidistant sampling.}

Furthermore, it may be advantageous in practice to characterize the detectability of a continuous-time system instead of its discretized representation (that depends on the chosen sampling scheme).
In particular, the proposed LMI conditions for computing an i-iIOSS Lyapunov function are in general less complex and thus easier to verify compared to \cite[Th.~2, Cor.~3]{Schiller2023c} applied to the discretized system dynamics.

The remainder of the paper is organized as follows. Section~\ref{sec:prelim} introduces the notation, system description, and basic definitions.
We present the MHE scheme in detail in Section~\ref{sec:MHE_setup} and derive sufficient conditions for guaranteed robust stability in Section~\ref{sec:MHE_RGAS}.
Section~\ref{sec:IOSS} provides sufficient LMI conditions to verify the underlying detectability assumption (i-iIOSS) {for special classes of nonlinear systems}.
We conclude in Section~\ref{sec:example} by illustrating the effectiveness of the proposed approach with a standard MHE benchmark example.

\section{Preliminaries}\label{sec:prelim}

\subsection{Notation}
Let $\mathbb{R}_{\geq0}$ denote the set of all non-negative real values.
By $|x|$, we indicate the Euclidean norm of the vector $x\in\mathbb{R}^n$. 
The quadratic norm with respect to a positive definite matrix $Q=Q^\top$ is denoted by {$|x|_Q^2=x^\top Q x$}, and the minimal and maximal eigenvalues of $Q$ are denoted by $\lambda_{\min}(Q)$ and $\lambda_{\max}(Q)$, respectively.
The maximum generalized eigenvalue of positive definite matrices $A=A^\top$ and $B=B^\top$ is denoted as $\lambda_{\max}(A,B)$, i.e., the largest scalar $\lambda$ satisfying $\det (A-\lambda B) = 0$.
For a measurable, locally essentially bounded function $z:\mathbb{R}_{\geq0}\rightarrow\mathbb{R}^n$, the essential sup-norm is defined as $\|z\| = \mathrm{ess}\sup_{t\geq0}|z(t)|$, and for the restriction of $z$ to an interval $[a,b]$ with $a,b\geq0$ as $\|z\|_{a:b} = \esssup_{t\in[a,b]}|z(t)|$.

Finally, we recall that a function $\alpha:\mathbb{R}_{\geq 0}\rightarrow\mathbb{R}_{\geq 0}$ is of class $\mathcal{K}$ if it is continuous, strictly increasing, and satisfies $\alpha(0)=0$; if additionally $\alpha(s)=\infty$ for $s\rightarrow\infty$, it is of class $\mathcal{K}_{\infty}$.
By $\mathcal{L}$, we refer to the class of functions $\theta:\mathbb{R}_{\geq 0}\rightarrow \mathbb{R}_{\geq 0}$ that are continuous, non-increasing, and satisfy $\lim_{t\rightarrow\infty}\theta(t)=0$, and by $\mathcal{KL}$ to the class of functions $\beta:\mathbb{R}_{\geq 0}\times\mathbb{R}_{\geq 0}\rightarrow\mathbb{R}_{\geq 0}$ with $\beta(\cdot,t)\in\mathcal{K}$ and $\beta(s,\cdot)\in\mathcal{L}$ for any fixed $t\in\mathbb{R}_{\geq 0}$ and $s\in\mathbb{R}_{\geq 0}$, respectively.

\subsection{Problem Setup}\label{sec:problem}
We consider the following system
\begin{subequations}\label{eq:sys}
	\begin{align}
	\dot{x}(t) &= f(x(t),u(t),w(t)),\label{eq:sys_1}\\
	y(t) &= h(x(t),u(t),w(t))\label{eq:sys_2}
	\end{align}
\end{subequations}
with states $x(t)\in\mathcal{X}\subseteq\mathbb{R}^n$, outputs $y(t)\in\mathcal{Y}\subseteq\mathbb{R}^p$, and time $t\geq0$.
The control and disturbance input signals $u$ and $w$ are measurable, locally essentially bounded functions\footnote{
	{See \cite[Appendix C.1]{Sontag1990} for definitions and further details on standard technical terms related to Lebesque measure theory.}
} taking values in $\mathcal{U}\subseteq\mathbb{R}^m$ and $\mathcal{W}\subseteq\mathbb{R}^q$, and we denote the set of such functions as $\mathcal{M}_{\mathcal{U}}$ and $\mathcal{M}_{\mathcal{W}}$, respectively.
The mappings $f:\mathcal{X}\times\mathcal{U}\times{\mathcal{W}}\rightarrow \mathcal{X}$ and $h:\mathcal{X}\times\mathcal{U}\times{\mathcal{W}}\rightarrow \mathcal{Y}$ constitute the system dynamics and output equation. In the following, we assume that $f$ and $h$ are jointly continuous and that $\mathcal{X}$ and $\mathcal{W}$ are closed.

Given an initial state $\chi\in\mathcal{X}$, we denote the solution of \eqref{eq:sys_1} at a time $t\geq0$ driven by the control $u\in\mathcal{M}_{\mathcal{U}}$ and the disturbance $w\in\mathcal{M}_{\mathcal{W}}$ by $x(t,\chi,u,w)$, and the corresponding output signal by $y(t,\chi,u,w):=h(x(t,\chi,u,w),u(t),w(t))$.
In the following, we assume that solutions of~\eqref{eq:sys} are unique and defined globally on $\mathbb{R}_{\geq0}$ for all $\chi\in\mathcal{X}$, $u\in\mathcal{M}_{\mathcal{U}}$, and $w\in\mathcal{M}_{\mathcal{W}}$.

To ensure that the unknown state trajectory of the system \eqref{eq:sys} can be (at least asymptotically) reconstructed, the system has to exhibit a suitable detectability property.
To this end, we employ the notion of (time-discounted) incremental integral input/output-to-state stability (i-iIOSS) as introduced in~\cite{Schiller2023a}---more precisely, its equivalent (\cite[Th.~1]{Schiller2023a}) Lyapunov function characterization.

\begin{definition}[i-iIOSS Lyapunov function, \protect{\cite[Def.~2]{Schiller2023a}}]
	\label{def:IOSS_Lyap}
	A function $U:\mathcal{X}\times\mathcal{X}\rightarrow\mathbb{R}_{\geq 0}$ is an \mbox{i-iIOSS} Lyapunov function for the system~\eqref{eq:sys} if it is continuous and there exist functions $\alpha_1,\alpha_2,\sigma_w,\sigma_{y}\in\mathcal{K}_{\infty}$ and a constant $\lambda\in[0,1)$ such that
	\begin{subequations}
		\label{eq:IOSS_Lyap}
		\begin{align}
		\label{eq:IOSS_Lyap_1}
		&\alpha_1(|\chi_1-\chi_2|)\leq U(\chi_1,\chi_2) \leq \alpha_2(|\chi_1-\chi_2|),\\[1ex]	
		&\ U(x_1(t),x_2(t)) \nonumber\\[-1ex]
		\leq &\ U(\chi_1,\chi_2)\lambda^t +  \int_{0}^{t} \lambda^{t-\tau}  \big(\sigma_{w}(|w_1(\tau)-w_2(\tau)|)\nonumber \\
		& \hspace{23ex} +\sigma_{y}(|y_1(\tau)-y_2(\tau)|)\big) d\tau \label{eq:IOSS_Lyap_2}
		\end{align}
	\end{subequations}
	for all $t\in\mathbb{R}_{\geq0}$, $\chi_1,\chi_2\in\mathcal{X}$, $u\in\mathcal{M}_{\mathcal{U}}$, and $w_1,w_2\in\mathcal{M}_{\mathcal{W}}$, where $x_1(\tau) = x(\tau,\chi_1,u,{w}_2)$, ${x}_2(\tau) = x(\tau,{\chi}_2,u,{w}_2)$ and 	$y_1(\tau) = y(\tau,\chi_1,u,{w}_2)$, ${y}_2(\tau) = y(\tau,{\chi}_2,u,{w}_2)$, $\tau\in[0,t]$.
\end{definition}

The concept of i-iIOSS was introduced in~\cite{Schiller2023a} as a continuous-time analogue of i-IOSS (in discrete time), which became the standard detectability assumption in the context of MHE/FIE in recent years, cf., e.g., \cite{Allan2021a,Schiller2023c,Knuefer2023,Hu2023,Rawlings2017}.
{This notion implies that the difference of any two system trajectories is upper bounded by the difference of their respective initial conditions, their inputs, and their outputs. Hence, if the differences of their inputs and outputs become (and stay) small, then eventually the difference of their states must also become small, which can be used in the sense of a detectability property.}
{Note that this is not a restrictive condition in the context of robustly stable state observers (which are the focus of this work), since i-iIOSS is in fact necessary for the existence of such, cf.~\cite[Prop.~2]{Schiller2023a}}.
{We also point out that the exponential decrease in \eqref{eq:IOSS_Lyap_2} is indeed without loss of generality, cf.~\cite[Rem.~1]{Schiller2023a}; this will be exploited in Section~\ref{sec:MHE} to construct MHE/FIE schemes for which robust stability of the estimation error can be guaranteed, cf. Remark~\ref{rem:disco}.}

Since we are interested in an optimization-based method for {(nonlinear)} state estimation, we inevitably have to employ some sort of sampling strategy.
To this end, let $\mathcal{T}\subset\mathbb{R}_{\geq0}$ be a set containing (arbitrary) distinct time instances.
Given some prior estimate $\hat{\chi}$ of the unknown initial condition $\chi$, the overall goal is to compute, at each sampling time $t_i\in\mathcal{T}$, the estimate $\hat{x}(t_i)$ of the true system state~$x(t_i)$.
In the next section, we provide conditions under which the resulting estimation error converges exponentially to a neighborhood around the origin by means of the following definition.

\begin{definition}[RGAS, RGES]
	\label{def:RGAS}
	A state estimator for system~\eqref{eq:sys} is robustly globally asymptotically stable (RGAS) if there exist functions $\beta,\beta_{x},\beta_{w}\in\mathcal{K}_{\infty}$ and a constant $\eta\in[0,1)$ such that the estimated state $\hat{x}$ with $\hat{x}(0)=\hat{\chi}$ satisfies
	\begin{align}	
	\label{eq:RGAS}	
	&\ \beta(|x(t_i,\chi,u,w)-\hat{x}(t_i)|) \nonumber\\
	\leq&\ \beta_{x}(|\chi-\hat{\chi}|)\eta^{t_i} + \int_{0}^{t_i}\eta^{t_i-\tau}\beta_{w}(|w(\tau)|)d\tau
	\end{align}
	for all $t_i\in\mathcal{T}$, all initial conditions $\chi,\hat{\chi}\in\mathcal{X}$, all controls $u\in\mathcal{M}_{\mathcal{U}}$, and all disturbances $w\in\mathcal{M}_{\mathcal{W}}$.
	If additionally $\beta(s)\geq C_1s^r$ and $\beta_x(s)\leq C_2s^r$ for some $C_1,C_2,r>0$, then the state estimator is robustly globally exponentially stable (RGES).
\end{definition}

\begin{remark}[Point-wise error bound]\label{rem:CT_RGAS}
	Note that we consider~\eqref{eq:RGAS} only point-wise for all $t_i\in\mathcal{T}$ (and not for all $t\geq0$), since the estimates are produced at certain time instances and not continuously, compare also~\cite{Michalska1995}.
	Extensions of the proposed MHE scheme to account for a pure continuous-time stability notion are possible by predicting the estimated state between two consecutive samples, e.g., using the nominal system dynamics or an additional auxiliary observer, cf., e.g.,~\cite{Schiller2022b}.
\end{remark}

The property~\eqref{eq:RGAS} provides a time-discounted $L^2$-to-$L^\infty$ bound for the estimation error.
It directly implies an \mbox{$L^\infty$-to-$L^\infty$} bound as shown in the following proposition.

\begin{proposition}\label{prop:RGES}
	Suppose there exists a state estimator for system~\eqref{eq:sys}. If it is RGAS in the sense of Definition~\ref{def:RGAS}, then there exist functions $\psi\in\mathcal{KL}$ and $\gamma\in\mathcal{K}_\infty$ such that
	\begin{equation}\label{eq:RGAS_sup}	
	|x(t_i,\chi,u,w)-\hat{x}(t_i)| \leq \max\{\psi(|\chi-\hat{\chi}|,{t_i}),\gamma(\|w\|_{0:t_i})\}
	\end{equation}
	for all $t_i\in\mathcal{T}$, $\chi,\hat{\chi}\in\mathcal{X}$, $u\in\mathcal{M}_{\mathcal{U}}$, and $w\in\mathcal{M}_{\mathcal{W}}$.
	If it is RGES, then there exist $C>0$ and $\rho\in[0,1)$ such that~\eqref{eq:RGAS_sup} holds with $\psi(s,t)$ replaced by $Cs\rho^t$ and a suitable function $\gamma\in\mathcal{K}_\infty$.
\end{proposition}

The proof is shifted to the Appendix, Section~\ref{sec:proof_prop_RGES}.

Overall, Definition~\ref{def:RGAS} is a powerful robust stability property of observers for several reasons. First, the integral term in~\eqref{eq:RGAS} can be viewed as the energy of the true disturbance signal $w$ under fading memory and thus has a reasonable physical interpretation, compare also~\cite{Praly1996}.
Second, it ensures a finite estimation error bound for both (unbounded) disturbance signals with finite energy (by~\eqref{eq:RGAS}) and persistent (non-vanishing) bounded disturbances (by application of Proposition~\ref{prop:RGES}), cf.~\cite{Schiller2023a} for more details.
Third, it is particularly useful for observers that do not admit a convenient state-space representation (such as MHE and FIE), since~\eqref{eq:RGAS} directly implies that $|x(t)-\hat{x}(t)|\rightarrow0$ if $|w(t)| \rightarrow 0$ for $t\rightarrow\infty$, compare also \cite{Allan2021} for a similar discussion in a discrete-time setting.

\section{Moving horizon estimation}\label{sec:MHE}

\subsection{Setup}\label{sec:MHE_setup}
In this section, we present the MHE scheme and provide conditions for guaranteed robust stability.
To this end, we restrict ourselves to exponential detectability rather than asymptotic detectability, i.e., the functions $\alpha_1,\alpha_2,\sigma_w,\sigma_{y}$ in \eqref{eq:IOSS_Lyap} are assumed to be of quadratic form.

\begin{assumption}[Exponential detectability]
	\label{ass:IOSS_Lyap}
	System~\eqref{eq:sys} admits a quadratically bounded i-iIOSS Lyapunov function~$U$ according to Definition~\ref{def:IOSS_Lyap} satisfying
	\begin{subequations}
		\label{eq:IOSS_Lyap_ass}
		\begin{align}
		\label{eq:IOSS_Lyap_1_ass}
		&|\chi_1-\chi_2|_{P_1}^2\leq U(\chi_1,\chi_2) \leq |\chi_1-\chi_2|_{P_2}^2,\\[1ex]	
		&\ U(x_1(t),x_2(t)) \nonumber\\[-1ex]
		\leq &\ U(\chi_1,\chi_2)\lambda^t +  \int_{0}^{t} \lambda^{t-\tau}  \big(|w_1(\tau)-{w}_2(\tau)|_Q^2\nonumber \\[-1ex]
		& \hspace{25ex} +|y_1(\tau)-{y}_2(\tau)|_R^2\big) d\tau, \label{eq:IOSS_Lyap_2_ass}
		\end{align}
	\end{subequations}
	with $P_1,P_2,Q,R\succ 0$ for all $t\geq0$, $\chi_1,{\chi}_2\in\mathcal{X}$, $u\in\mathcal{M}_{\mathcal{U}}$, and $w_1,{w}_2\in\mathcal{M}_{\mathcal{W}}$.
\end{assumption}

Note that the requirement of exponential detectability is a standard condition for robust stability of MHE, ensuring a linear contraction of the estimation error over a (finite) horizon, cf., e.g., \cite{Allan2021a,Schiller2023c,Knuefer2023,Hu2023,Rawlings2017}.
{Moreover, using the same reasoning as in the proofs of \cite[Th.~1, Prop.~2]{Schiller2023a}, one can show that Assumption~\ref{ass:IOSS_Lyap} is a necessary condition for the existence of RGES observers.}
In Section~\ref{sec:IOSS}, we provide sufficient conditions for constructing a quadratic i-iIOSS Lyapunov function $U$ satisfying \eqref{eq:IOSS_Lyap_ass}.

Let the initial estimate $\hat{\chi}\in\mathcal{X}$ be given. At any sampling time $t_i\in\mathcal{T}$, the proposed MHE scheme considers the past input and output trajectories of the system~\eqref{eq:sys} within the moving time interval $[t_i-T_{t_i},t_i]$ of length $T_{t_i} = \min\{t_i, T\}$ for some $T>0$.
Let $u_{t_i} : [0,T_{t_i}) \rightarrow \mathcal{M}_{\mathcal{U}}$ and $y_{t_i} : [0,T_{t_i}) \rightarrow \mathcal{M}_{\mathcal{Y}}$ denote the currently involved segments of the input and output trajectories of system~\eqref{eq:sys}, which are defined as
\begin{align}
u_{t_i}(\tau) &:= u(t_i-T_{t_i}+\tau),\ \tau\in[0,T_{t_i}),\\
y_{t_i}(\tau) &:= y(t_i-T_{t_i}+\tau,\chi,u,w),\ \tau\in[0,T_{t_i}).
\end{align}
Then, the optimal state trajectory on the interval $[t_i-T_{t_i},t_i]$ is obtained by solving the following optimization problem:
\begin{subequations}\label{eq:MHE}
	\begin{align}\label{eq:MHE_0}
	&\min_{\bar{\chi}_{t_i},\bar{w}_{t_i}}
	J(\bar{\chi}_{t_i},\bar{w}_{t_i},\bar{y}_{t_i},{t_i}) \\ 
	\text{s. t. }
	&\bar{x}_{t_i}(\tau)=x(\tau,\bar{\chi}_{t_i},u_{{t_i}},\bar{w}_{t_i}),\ \tau\in[0,T_{t_i}], \label{eq:MHE_1} \\	
	&\bar{x}_{t_i}(\tau)\in\mathcal{X},\ \tau\in[0,T_{t_i}] \\		
	&\bar{y}_{t_i}(\tau)=y(\tau,\bar{\chi}_{t_i},u_{{t_i}},\bar{w}_{t_i}),\ \tau\in[0,T_{t_i}), \label{eq:MHE_2} \\	
	&\bar{w}_{t_i}(\tau)\in{\mathcal{W}}, \ \bar{y}_{t_i}(\tau)\in{\mathcal{Y}},\	\tau\in[0,T_{t_i}).\
	\label{eq:MHE_3}
	\end{align}
\end{subequations}
The decision variables $\bar{\chi}_{t_i}$ and $\bar{w}_{t_i} : [0,T_{t_i})\rightarrow\mathcal{W}$ denote the estimates of the state at the beginning of the horizon and the disturbance signal over the horizon, respectively, estimated at time $t_i$.
Given the input trajectory segment $u_{t_i}$, these decision variables (uniquely) determine the estimated state and output trajectories $\bar{x}_{t_i}$ and $\bar{y}_{t_i}$ via~\eqref{eq:MHE_1} and~\eqref{eq:MHE_2} as functions defined on $[0,T_{t_i}]$ and $[0,T_{t_i})$, respectively.
In~\eqref{eq:MHE_0}, we consider the discounted objective
\begin{align}
&\ J(\bar{\chi}_{t_i},\bar{w}_{t_i},\bar{y}_{t_i},{t_i})\nonumber \\
=&\ \Gamma(\bar{\chi}_{t_i},\hat{x}({t_i}-T_{t_i}))\lambda^{T_{t_i}} \nonumber \\
&\ + \int_{0}^{{T_{t_i}}}\lambda^{T_{t_i}-j} L(\bar{w}_{t_i}(\tau),y_{t_i}(\tau)-\bar{y}_{t_i}(\tau))d\tau \label{eq:MHE_objective}
\end{align}
with quadratic prior weighting $\Gamma(\chi,x) = 2|\chi-x|_{P_2}^2$ and quadratic stage cost $L(w,\Delta y) = 2|w|^2_Q + |\Delta y|_R^2$.
Here, $\hat{x}(t_i-T_{t_i})$ is the prior estimate and the parameters $P_2,Q,R,\lambda$ are from Assumption~\ref{ass:IOSS_Lyap}.

\begin{remark}[Discounting]\label{rem:disco}
	The use of exponential discounting in~\eqref{eq:MHE_objective} establishes a direct link between the detectability property (Assumption~\ref{ass:IOSS_Lyap}), the MHE scheme (via the cost function~\eqref{eq:MHE_objective}), and desired stability property (RGES, see Theorem~\ref{thm:MHE} below).
	It is motivated by recent time-discounted MHE approaches for discrete-time systems, which, compared to their non-discounted counterparts, allow the direct derivation of stronger and less conservative robustness guarantees by leveraging this particular structural connection, cf.~\cite{Knuefer2018,Knuefer2023,Hu2023,Schiller2023c}.
\end{remark}

\begin{remark}[Tuning]\label{rem:cost}
	We point out that fixing the weighting matrices $P_2,Q,R$ in the MHE objective~\eqref{eq:MHE_objective} to the values from the i-iIOSS Lyapunov function is in fact without loss of generality and therefore does not restrict any tuning possibilities.
	In particular, the scaled Lyapunov function $\tilde{U}(x_1,x_2) = KU(x_1,x_2)$ with
	\begin{equation}\label{eq:def_K}
	K := \left(\max\{\lambda_{\max}(P_2,\tilde{P}_2),\lambda_{\max}(Q,\tilde{Q}),\lambda_{\max}(R,\tilde{R})\}\right)^{-1}
	\end{equation}
	satisfies\footnote{
		Generally, note here that $|x|_A^2\leq\lambda_{\max}(A,B)|x|_B^2$ for any real vector $x$ and any $A,B\succ0$ of appropriate dimensions.
	} Assumption~\ref{ass:IOSS_Lyap} with $P_2,Q,R$ replaced by arbitrary positive definite matrices $\tilde{P}_2,\tilde{Q},\tilde{R}$ and $P_1$ replaced by $\tilde{P}_1=KP_1$, compare also~\cite[Rem.~1]{Schiller2023c} for a similar discussion in a discrete-time setting.
	Note that this also allows for choosing the weights in~\eqref{eq:MHE_objective} time-varying (e.g., based on Kalman-filtering-like update recursions) if uniform bounds are either known \textit{a priori} or imposed online.
	
	Moreover, we point out that $\lambda$ in~\eqref{eq:MHE_objective} can be replaced by any $\tilde{\lambda}\in[\lambda,1)$, since the dissipation inequality~\eqref{eq:IOSS_Lyap_2_ass} for a given i-iIOSS Lyapunov function $U$ is still a valid dissipation inequality for $U$ if $\lambda$ is replaced by $\tilde{\lambda}\in[\lambda,1)$.
	Throughout this paper, we choose the parameters $\lambda,P_2,Q,R$ in~\eqref{eq:MHE_objective} according to those from the i-iIOSS Lyapunov function from Assumption~\ref{ass:IOSS_Lyap} to simplify the notation.
\end{remark}

Note that the optimization problem~\eqref{eq:MHE} and \eqref{eq:MHE_objective} admits a global solution under Assumption~\ref{ass:IOSS_Lyap}, since $\Gamma$ and $L$ are positive definite (due to positive definiteness of $P_2,Q,R$) and radially unbounded in the decision variables and the sets $\mathcal{X}$ and $\mathcal{W}$ are closed.
We denote a minimizer to~\eqref{eq:MHE} and \eqref{eq:MHE_objective} by $(\bar{\chi}^*_{t_i},\bar{w}_{t_i}^*)$, and the corresponding optimal state trajectory by $\bar{x}^*_{t_i}(\tau) = x(\tau,\bar{\chi}^*_{t_i},u_{t_i},\bar{w}_{t_i}^*)$, $\tau\in[0,T_{t_i}]$.
The resulting state estimate at sampling instance $t_i\in\mathcal{T}$ is then given by
\begin{equation}
\hat{x}(t_i) := \bar{x}_{t_i}^*(T_{t_i}). \label{eq:x_hat}
\end{equation}
Furthermore, we define the estimated state trajectory $\hat{x}(t)$, $t\in[0,t_i]$ as the piece-wise continuous function resulting from the concatenation of optimal trajectory segments according to
\begin{equation} \label{eq:x_hat_CT} 
\hat{x}(t) := 
\begin{cases}
\bar{x}_{k(t)}^*(t-k(t)+T_{k(t)}), &  t\in(0,t_i]\\
\hat{\chi} & t=0
\end{cases}
\end{equation}
for all $t_i\in\mathcal{T}$, where $k(t)$ is the sampling time associated with $t$ defined by
\begin{equation}\label{eq:def_k}
k(t):=\min_{k\in\{k\in\mathcal{T}:k\geq t\}}k.
\end{equation}

\begin{remark}[Estimated trajectory]\label{rem:x_CT}
	Computing the piecewise continuous state trajectory~\eqref{eq:x_hat_CT} is essential to ensure that the prior estimate $\hat{x}(t_i-T_{t_i})$ in the MHE objective in~\eqref{eq:MHE_objective} is well-defined for all $t_i\in\mathcal{T}$.
	To ensure that it is available at time $t_i\in\mathcal{T}$ (i.e., has already been computed in the past), the horizon length $T$ must naturally satisfy
	\begin{equation}\label{eq:def_delta}
	T>\bar{\delta}:=\sup\limits_{t\geq 0} k(t)-t,
	\end{equation}
	where $\bar{\delta}$ can be referred to as the maximum deviation between a time $t$ and its associated sampling time $k(t)\geq t$ that may occur for all $t\geq0$.
	For equidistant sampling using a constant sampling period $\delta>0$, i.e., when $\mathcal{T}=\{t\in\mathbb{R}_{\geq0}: t = n\delta, n\in\mathbb{N}_{\geq0}\}$, it trivially holds that $\bar{\delta}=\delta$.
	Note that computing the full trajectory~\eqref{eq:x_hat_CT} (and hence the third step in the algorithm below) could be avoided by designing $\mathcal{T}$ such that $t_i-T_{t_i}\in\mathcal{T}$ for all $t_i\in\mathcal{T}$, compare also Remark~\ref{rem:T}.
\end{remark}

The MHE problem \eqref{eq:MHE} is solved in a receding horizon fashion, and the corresponding algorithm can be summarized as follows.
Given the sampling time $t_i\in\mathcal{T}$ and its predecessor $t_i^-\in\mathcal{T}$ (if there is none, set $t_i^-=0$),
\begin{enumerate}
	\item collect the input and output trajectory segments $u_{t_i}$ and $y_{t_i}$,
	\item solve the MHE problem~\eqref{eq:MHE} with objective \eqref{eq:MHE_objective},
	\item update the estimated trajectory~\eqref{eq:x_hat_CT} by attaching the most recent optimal trajectory segment $\bar{x}_{t_i}^*(\tau)$, $\tau\in(t_i^-,t_i]$,
	\item update $t_i^- = t_i$, pick the next sampling time $t_{i} = \min_{k\in\{k\in\mathcal{T}:k> t_i\}}k$, and go back to 1.
\end{enumerate}

The stability properties of the proposed continuous-time MHE are established in the next section. In particular, we provide simple conditions for designing a suitable set $\mathcal{T}$ and a horizon length $T$ such that RGES of MHE (cf.~Definition~\ref{def:RGAS}) is guaranteed.

\subsection{Robust stability analysis}\label{sec:MHE_RGAS}

In the following, we show how the i-iIOSS Lyapunov function $U$ from Assumption~\ref{ass:IOSS_Lyap} can be used to characterize a decrease of the estimation error $x(t)-\hat{x}(t)$ in Lyapunov coordinates over the interval $[t_i-T_{t_i},t]$ for any $t\geq0$ and its corresponding sampling time $t_i=k(t)$ (where we define $x(t):=x(t,\chi,u,w)$, $t\geq0$ for notational brevity).
Indeed, this requires invoking the specific choice of the cost function~\eqref{eq:MHE_objective} involving the optimal trajectories over the interval $[t_i-T_{t_i},t_i]$ estimated at time $t_i$.
Then, we apply this bound recursively to establish RGES of MHE in the sense of Definition~\ref{def:RGAS}, see Theorem~\ref{thm:MHE} below.

\begin{proposition}\label{prop:MHE}
	Let Assumption~\ref{ass:IOSS_Lyap} hold. Then, the state estimate $\hat{x}(t)$ from~\eqref{eq:x_hat_CT} satisfies
	\begin{align}
	&\ U(x(t),\hat{x}(t)) \nonumber \\
	\leq &\ \lambda^{-(t_i-t)}\Big(4\lambda_{\max}(P_2,P_1)\lambda^{T_{t_i}}U(x(t_i-T_{t_i}),\hat{x}(t_i-T_{t_i}))\nonumber \\
	&\hspace{11ex} +  4 \int_{t_i-T_{t_i}}^{t_i}\lambda^{t_i-\tau}|w(\tau)|_Q^2 d\tau\Big), \label{eq:thm_1}
	\end{align}
	for all $t\geq0$, $\chi,\hat{\chi}\in\mathcal{X}$, $u\in\mathcal{M}_{\mathcal{U}}$, $w\in\mathcal{M}_{\mathcal{W}}$, where $t_i=k(t)$ with $k(t)$ from~\eqref{eq:def_k} and $\lambda = e^{-\kappa}$.
\end{proposition}

The proof is shifted to the Appendix, Section~\ref{sec:proof_prop_MHE}; the main idea is similar to that of \cite[Prop~1]{Schiller2023c}, with technical differences due to the continuous-time setup.

In the following, we consider the case when the horizon $T$ and set of sampling times $\mathcal{T}$ are designed such that
\begin{equation}\label{eq:cond_T}
4\lambda_{\max}(P_2,P_1)\lambda^{T-\bar{\delta}} =: \rho^{T-\bar{\delta}} \in(0,1) 
\end{equation}
holds for some $\rho\in(0,1)$, where $\bar{\delta}$ is defined in~\eqref{eq:def_delta}.
Provided that an i-iIOSS Lyapunov function as in Assumption~\ref{ass:IOSS_Lyap} is known, we point out that for any design of~$\mathcal{T}$, condition~\eqref{eq:cond_T} can be easily satisfied by choosing $T$ such that
\begin{equation}\label{eq:def_T}
T> -\frac{\ln(4\lambda_{\max}(P_2,P_1))}{\ln(\lambda)}+\bar{\delta},
\end{equation}
leading to
\begin{equation}
\rho=(4\lambda_{\max}(P_2,P_1))^{\frac{1}{T-\bar{\delta}}}\lambda.
\end{equation}
Whenever \eqref{eq:cond_T} is satisfied, Proposition~\ref{prop:MHE} directly provides a dissipation inequality in integral form with exponential decrease for the estimation error in Lyapunov coordinates; consequently, the i-iIOSS Lyapunov function $U$ can be viewed as a Lyapunov-like function for MHE on each interval $[k(t)-T_{k(t)},t]$ for all $t\geq0$.

In the following, we establish RGES of MHE by applying the bound~\eqref{eq:thm_1} recursively to cover the whole interval $[0,t_i]$ for a given sampling time $t_i\in\mathcal{T}$.
However, special care must be taken when concatenating the dissipation inequalities due to the overlap of their domains.

\begin{theorem}[RGES of MHE]\label{thm:MHE}
	Let Assumption~\ref{ass:IOSS_Lyap} hold. Suppose that the horizon length $T$ is chosen such that~\eqref{eq:cond_T} is satisfied for some $\rho\in(0,1)$.
	Then, the estimation error satisfies
	\begin{align}
	|x(t_i)-\hat{x}(t_i)|_{P_1}^2
	\leq 4\rho^{t_i}|\chi-\hat{\chi}|^2_{P_2} + 8\int_{0}^{t_i}\rho^{t_i-\tau}|w(\tau)|_{Q}^2 d \tau \label{eq:thm_2}
	\end{align}
	for all $t_i\in\mathcal{T}$ and all $\chi,\hat{\chi}\in\mathcal{X}$, $u\in\mathcal{M}_{\mathcal{U}}$, and $w\in\mathcal{M}_{\mathcal{W}}$, i.e., the proposed MHE scheme is RGES in the sense of Definition~\ref{def:RGAS}.
\end{theorem}
\begin{proof}
	We start by noting that condition~\eqref{eq:cond_T} implies $\lambda\leq\rho$ and
	\begin{equation}
	4\lambda_{\max}(P_2,P_1)\lambda^{T} = \rho^{T-\bar{\delta}}\lambda^{\bar{\delta}}. \label{eq:proof_ind_00}
	\end{equation}
	In the following, we will make use of the relation
	\begin{align}
	&\ 4\lambda_{\max}(P_2,P_1)\lambda^{T-(k(t)-t)}
	\stackrel{\eqref{eq:proof_ind_00}}{=}\rho^{T-\bar{\delta}}\lambda^{\bar{\delta}-(k(t)-t)}\nonumber\\
	\leq&\ \rho^{T-\bar{\delta}}\rho^{\bar{\delta}-(k(t)-t)} =\rho^{T-(k(t)-t)}, \label{eq:proof_cond_T}
	\end{align}
	which holds for any $t\geq0$.
	Assume that an arbitrary sampling time $t_i\in\mathcal{T}$ is given.
	We define the sequence of sampling times $\{k_{j}\}, j\in\mathbb{N}_0$ starting at $k_0=t_i$ using the following recursion:
	\begin{equation}
	k_{j+1}	= \min_{k\in\mathcal{K}_j} k \label{eq:proof_sampling}
	\end{equation}
	with
	\begin{equation*}
	\mathcal{K}_j := \{k\in\mathcal{T}:k_j > k\geq k_j-T\}
	\end{equation*}
	if the set $\mathcal{K}_j$ is non-empty, and $k_{j+1}=k_j$ otherwise.
	In the following, we use $c:=4\lambda_{\max}(P_2,P_1)$ for notational brevity.
	
	Now assume that for some time $s\geq0$, $k_j = k(s)\geq T$ is the corresponding sampling instance.
	From Proposition~\ref{prop:MHE}, it follows that
	\begin{align}
	&\ U(x(s),\hat{x}(s)) \nonumber\\
	\leq&\ \lambda^{-(k_{j}-s)}\Big(c\lambda^{T}U(x(k_{j}-T),\hat{x}(k_{j}-T))\nonumber \\
	&\ +4\int_{s}^{k_{j}}\lambda^{k_{j}-\tau}|w(\tau)|_{Q}^2 d \tau  + 4\int_{k_{j+1}}^{s}\lambda^{k_{j}-\tau}|w(\tau)|_{Q}^2 d \tau\nonumber\\
	&\
	+ 4\int_{k_{j}-T}^{k_{j+1}}\lambda^{k_{j}-\tau}|w(\tau)|_{Q}^2 d \tau \Big).	\label{eq:proof_ind_0}
	\end{align}
	We aim to apply this property recursively.
	To this end, we have split the integral on the right-hand side involving the interval $[k_{j}-T,k_{j}]$ in three parts (noting that $k_{j}-T \leq k_{j+1} < s \leq k_{j}$); the first part overlaps with the previous iteration covering the interval $[k_{j-1}-T,k_{j-1}]$ (unless $k_{j-1}-T\in\mathcal{T} \Rightarrow k_j=k_{j-1}-T$), and the third part overlaps with the succeeding iteration covering the interval $[k_{j+1}-T,k_{j+1}]$ (unless $k_{j}-T\in\mathcal{T} \Rightarrow k_{j+1} = k_{j}-T$).
	
	We claim that
	\begin{align}
	&\ U(x(k_0),\hat{x}(k_0))\nonumber\\
	\leq&\ c\lambda^{T}\rho^{k_0-k_{j}}U(x(k_{j}-T),\hat{x}(k_{j}-T)) \nonumber\\
	&\ +8\int_{k_{j+1}}^{k_0}\rho^{k_0-\tau}|w(\tau)|_{Q}^2 d \tau
	+ 4\int_{k_{j}-T}^{k_{j+1}}\rho^{k_0-\tau}|w(\tau)|_{Q}^2 d \tau\label{eq:proof_ind}
	\end{align}
	holds for any $k_0\in\mathcal{T}$ and all $j\in\mathbb{N}_0$ for which $k_{j+1}\geq T$ is satisfied, and we give a proof by induction.
	For the base case, consider~\eqref{eq:proof_ind_0} with $s=k_0$. We obtain
	\begin{align*}
	&\ U(x(k_0),\hat{x}(k_0))\\
	\leq&\ c\lambda^{T}U(x(k_0-T),\hat{x}(k_0-T)) \\
	&\
	+ 4\int_{k_{1}}^{k_0}\lambda^{k_0-\tau}|w(\tau)|_{Q}^2 d \tau
	+ 4\int_{k_0-T}^{k_{1}}\lambda^{k_0-\tau}|w(\tau)|_{Q}^2 d \tau,	
	\end{align*}
	for which~\eqref{eq:proof_ind} with $j=0$ serves as an upper bound.	
	
	We now prove~\eqref{eq:proof_ind} for general integers $j\in\mathbb{N}_0$.
	To this end, consider~\eqref{eq:proof_ind_0} with $s=k_j-T$. If $j\in\mathbb{N}_0$ is such that $k_{j+1}\geq T$, we can use the fact that
	\begin{align}
	&\ U(x(k_j-T),\hat{x}(k_j-T))\nonumber\\
	\leq&\ 	
	\lambda^{-(k_{j+1}-(k_{j}-T))}\Big(c\lambda^{T}U(x(k_{j+1}-T),\hat{x}(k_{j+1}-T)) \nonumber\\
	&\ +4\int_{k_j-T}^{k_{j+1}}\lambda^{k_{j+1}-\tau}|w(\tau)|_{Q}^2 d \tau\nonumber\\
	&\ + 4\int_{k_{j+2}}^{k_j-T}\lambda^{k_{j+1}-\tau}|w(\tau)|_{Q}^2 d \tau\nonumber\\
	&\ + 4\int_{k_{j+1}-T}^{k_{j+2}}\lambda^{k_{j+1}-\tau}|w(\tau)|_{Q}^2 d \tau\Big).\label{eq:proof_ind_1}
	\end{align}
	Now assume that~\eqref{eq:proof_ind} is true for some integer $j\in \mathbb{N}_0$ for which $k_{j+1} \geq T$. In the following, we show that~\eqref{eq:proof_ind} then also holds for $j+1$.
	The combination of~\eqref{eq:proof_ind} and~\eqref{eq:proof_ind_1} yields
	\begin{align}
	&\ U(x(k_0),\hat{x}(k_0))\nonumber\\
	\leq&\ c\lambda^{T-(k_{j+1}-(k_{j}-T))}\rho^{k_0-k_{j}}\nonumber\\
	&\ \cdot \Big(c\lambda^{T}U(x(k_{j+1}-T),\hat{x}(k_{j+1}-T))\nonumber \\
	&\qquad +4\int_{k_{j}-T}^{k_{j+1}}\lambda^{k_{j+1}-\tau}|w(\tau)|_{Q}^2 d \tau\nonumber\\
	&\qquad + 4\int_{k_{j+2}}^{k_{j}-T}\lambda^{k_{j+1}-\tau}|w(\tau)|_{Q}^2 d \tau\nonumber\\
	&\qquad + 4\int_{k_{j+1}-T}^{k_{j+2}}\lambda^{k_{j+1}-\tau}|w(\tau)|_{Q}^2 d \tau	\Big)\nonumber\\
	&\qquad +8\int_{k_{j+1}}^{k_0}\rho^{k_0-\tau}|w(\tau)|_{Q}^2 d \tau
	+ 4\int_{k_{j}-T}^{k_{j+1}}\rho^{k_0-\tau}|w(\tau)|_{Q}^2 d \tau \label{eq:proof_ind_2}
	\end{align}
	From \eqref{eq:proof_cond_T} and \eqref{eq:proof_sampling}, we can infer that
	\begin{align}
	c\lambda^{T-(k_{j+1}-(k_{j}-T))}\rho^{k_0-k_{j}} &\leq \rho^{T-(k_{j+1}-(k_{j}-T))}\rho^{k_0-k_{j}} \nonumber \\
	&= \rho^{k_0-k_{j+1}}. \label{eq:proof_ind_3}
	\end{align}
	Applying \eqref{eq:proof_ind_3} to \eqref{eq:proof_ind_2} and using that $\lambda\leq\rho$ leads to
	\begin{align*}	
	&\ U(x(k_0),\hat{x}(k_0))\\
	\leq&\ \rho^{k_0-k_{j+1}} c\lambda^{T}U(x(k_{j+1}-T),\hat{x}(k_{j+1}-T))  \\
	&\ +4\int_{k_{j}-T}^{k_{j+1}}\rho^{k_0-\tau}|w(\tau)|_{Q}^2 d \tau
	+ 4\int_{k_{j+2}}^{k_{j}-T}\rho^{k_0-\tau}|w(\tau)|_{Q}^2 d \tau \\
	&\ + 4\int_{k_{j+1}-T}^{k_{j+2}}\rho^{k_0-\tau}|w(\tau)|_{Q}^2 d \tau	\\
	&\ +8\int_{k_{j+1}}^{k_0}\rho^{k_0-\tau}|w(\tau)|_{Q}^2 d \tau
	+ 4\int_{k_{j}-T}^{k_{j+1}}\rho^{k_0-\tau}|w(\tau)|_{Q}^2 d \tau\\
	\leq&\ \rho^{k_0-k_{j+1}} c\lambda^{T}U(x(k_{j+1}-T),\hat{x}(k_{j+1}-T)) \\
	&\ + 8\int_{k_{j+2}}^{k_0}\rho^{k_0-\tau}|w(\tau)|_{Q}^2 d \tau
	+ 4\int_{k_{j+1}-T}^{k_{j+2}}\rho^{k_0-\tau}|w(\tau)|_{Q}^2 d \tau.
	\end{align*}		
	This proves~\eqref{eq:proof_ind} for all $k_0\in\mathcal{T}$ and all $j\in\mathbb{N}_0$ for which $k_{j+1}\geq T$ is satisfied.
	In fact, the above argument shows that~\eqref{eq:proof_ind} also holds for the smallest $j\in \mathbb{N}_0$ for which $k_{j+1} < T$.
	Furthermore, when $k_{j+1}<T$, from \eqref{eq:proof_ind_end} it follows that
	\begin{align*}
	&\ U(x(k_{j}-T),\hat{x}(k_{j}-T))\\
	\leq &\
	\lambda^{-(k_{j+1}-(k_{j}-T))}\Big(4\lambda^{k_{j+1}}|x(0)-\hat{x}(0)|^2_{P_2} \nonumber\\
	&\ + 4\int_{k_{j}-T}^{k_{j+1}}\lambda^{k_{j+1}-\tau}|w(\tau)|_{Q}^2 d \tau\\
	&\ + 4\int_{0}^{k_{j}-T}\lambda^{k_{j+1}-\tau}|w(\tau)|_{Q}^2 d \tau \Big).
	\end{align*}
	The combination with~\eqref{eq:proof_ind} leads to
	\begin{align}
	&\hspace{3ex} U(x(k_0),\hat{x}(k_0))\nonumber\\
	&\leq c\lambda^{T-(k_{j+1}-(k_{j}-T))}\rho^{k_0-k_{j}}		\Big(4\lambda^{k_{j+1}}|x(0)-\hat{x}(0)|^2_{P_2} \nonumber\\
	&\ +4\hspace{-0.5ex}\int_{k_{j}-T}^{k_{j+1}}\lambda^{k_{j+1}-\tau}|w(\tau)|_{Q}^2 d \tau + 4\hspace{-0.5ex}\int_{0}^{k_{j}-T}\lambda^{k_{j+1}-\tau}|w(\tau)|_{Q}^2 d \tau \Big) \nonumber\\
	&\  +8\int_{k_{j+1}}^{k_{0}}\rho^{k_0-\tau}|w(\tau)|_{Q}^2 d \tau
	+ 4\int_{k_{j}-T}^{k_{j+1}}\rho^{k_0-\tau}|w(\tau)|_{Q}^2 d \tau\label{eq:proof_ind_4}
	\end{align}	
	By using~\eqref{eq:proof_ind_3} and the fact that $\lambda\leq\rho$, from \eqref{eq:proof_ind_4} it follows that
	\begin{align}
	&\ U(x(k_0),\hat{x}(k_0))\nonumber\\
	\leq&\ 4\rho^{k_0}|x(0)-\hat{x}(0)|^2_{P_2} \nonumber\\
	&\ +4\int_{k_{j}-T}^{k_{j+1}}\rho^{k_0-\tau}|w(\tau)|_{Q}^2 d \tau
	+ 4\int_{0}^{k_{j}-T}\rho^{k_0-\tau}|w(\tau)|_{Q}^2 d \tau \nonumber\\
	&\ +8\int_{k_{j+1}}^{k_{0}}\rho^{k_0-\tau}|w(\tau)|_{Q}^2 d \tau
	+ 4\int_{k_{j}-T}^{k_{j+1}}\rho^{k_0-\tau}|w(\tau)|_{Q}^2 d \tau\nonumber\\
	\leq&\ 4\rho^{k_0}|x(0)-\hat{x}(0)|^2_{P_2} 
	+ 8\int_{0}^{k_{0}}\rho^{k_0-\tau}|w(\tau)|_{Q}^2 d \tau. \label{eq:proof_ind_sum}
	\end{align}
	Applying the lower bound in~\eqref{eq:IOSS_Lyap_1_ass} and recalling that $t_i=k_0$ yields \eqref{eq:thm_2}, which finishes this proof.
\end{proof}

\begin{remark}[Sampling strategy]\label{rem:T}	
	If $\mathcal{T}$ is such that $t_i-T_{t_i}\in\mathcal{T}$ for all $t_i\in\mathcal{T}$ (which is the case, e.g., for equidistant sampling with a constant period $\delta>0$ and $T$ being an integer multiple of $\delta$, cf. Remark~\ref{rem:x_CT}), it follows that the interval boundaries of each time horizon considered in the MHE optimization problem~\eqref{eq:MHE} coincide exactly with two sampling times.
	Consequently, it suffices to evaluate~\eqref{eq:thm_1} only at sampling times $t=t_i\in\mathcal{T}$---yielding a more direct Lyapunov-like function for MHE---and simply apply this bound recursively to establish RGES as is the case in discrete-time settings (e.g., \cite[Cor.~1]{Schiller2023c}), in particular without the need to take into account any discrepancy between horizon boundaries and sampling times.
	These simplifications result in a slightly weaker condition on the horizon length and a slightly improved RGES result; namely, we can take $\bar{\delta}=0$ in \eqref{eq:cond_T} and replace the factor~$8$ in~\eqref{eq:thm_2} by the factor $4$.
\end{remark}

\begin{remark}[Full information estimation]
	The restriction to exponential detectability (Assumption~\ref{ass:IOSS_Lyap}) can be relaxed to asymptotic detectability (Definition~\ref{def:IOSS_Lyap}) if FIE is applied {(that is, the MHE scheme~\eqref{eq:MHE} with $T_{t_i}=t_i$)}.
	In this case, we consider the cost function~\eqref{eq:MHE_objective} with $T_{t_i}$ replaced by~$t_i$, $\Gamma(\chi,x)=\alpha_2(2|\chi-x|)$ and $L(w,\Delta y) = \sigma_w(2|w|) + \sigma_y(|\Delta y|)$.
	By applying the same steps as in the proof of Proposition~\ref{prop:MHE}, we can infer that the corresponding estimation error satisfies	
	\begin{align*}
	\alpha_1(|x(t_i)-\hat{x}(t_i)|) \leq&\ 2\lambda^{t_i}\alpha_2(2|\chi-\hat{\chi}|)\\
	&\ + 2\int_{0}^{t_i}\lambda^{t_i-\tau}\sigma_w(2|w(\tau)|) d \tau
	\end{align*}
	for all $t_i\in\mathcal{T}$ and all $\chi,\hat{\chi}\in\mathcal{X}$, $u\in\mathcal{M}_{\mathcal{U}}$, and $w\in\mathcal{M}_{\mathcal{W}}$, which implies RGAS of FIE in the sense of Definition~\ref{def:RGAS}.
\end{remark}

\section{Constructing smooth i-iIOSS Lyapunov functions}\label{sec:IOSS}

The construction of i-iIOSS Lyapunov functions satisfying Assumption~\ref{ass:IOSS_Lyap} can be performed similarly to the discrete-time case in \cite[Sec.~IV]{Schiller2023c}---namely, by applying arguments from contraction theory and Riemannian geometry. 
Such a differential approach offers significant advantages, since the problem can be reformulated using LMI conditions that can be efficiently verified using standard tools such as sum-of-squares (SOS) optimization or linear parameter-varying (LPV) embeddings.
Therefore, it is not surprising that this method has numerous applications in the analysis and design of controller and observers for nonlinear systems, cf., e.g., \cite{Forni2014,Manchester2018,Yi2021,Manchester2017,Singh2017,Lopez2021}.

In the remainder of this section, for ease of presentation, we focus on the construction of quadratic i-iIOSS Lyapunov functions only and impose the following regularity properties on the system and its input signals.
\begin{assumption}\label{ass:IOSS}
	The function $f$ is continuously differentiable in all of its arguments, $h$ is affine in $(x,w)$, and the sets $\mathcal{X}$ and $\mathcal{W}$ are convex.
	Furthermore, the input signals $u$ and $w$ are piecewise\footnote{
		If the following results are to be applied to the methods presented in Section~\ref{sec:MHE}, it must be ensured that the estimated disturbance $\bar{w}_{t_i}$ in the MHE problem in~\eqref{eq:MHE} is also piecewise right-continuous. In practice, this is immediately the case when standard discretization methods are used to solve~\eqref{eq:MHE}.
	} right-continuous.
\end{assumption}

{Affinity of $h$ is a technical requirement in the proof of Theorem~\ref{thm:IOSS} (compare Section~\ref{sec:proof_thm_IOSS} in the Appendix) to easily transfer a differential property to an incremental one (it basically ensures that the output $h$ evaluated along the shortest path between any two points $x_1,x_2\in\mathcal{X}$ is a linear combination of the respective outputs $y_1,y_2$ for any point of that path).
	It is directly satisfied if a subset (or linear combination) of the system state is measured, which is the case for many practical applications. Furthermore, assuming affinity of $h$ is also quite general in the sense that it covers several observable canonical forms and hence models that admit a suitable transformation, compare also \cite[Sec.~IV]{Schiller2023c}, \cite[Rem.~1]{Yi2021}.}

We define the linearizations of~\eqref{eq:sys} at a given point $(x,u,w) \in \mathcal{X}\times\mathcal{U}\times\mathcal{W}$ as
\begin{equation}\label{eq:IOSS_ABCD}
\begin{matrix}
A=\displaystyle\frac{\partial f}{\partial x}(x,u,w),\
B=\displaystyle\frac{\partial f}{\partial w}(x,u,w),\\[2ex]
C=\displaystyle\frac{\partial h}{\partial x}(x,u,w),\
D=\displaystyle\frac{\partial h}{\partial w}(x,u,w).
\end{matrix}
\end{equation}

\begin{theorem}\label{thm:IOSS}
	Let Assumption~\ref{ass:IOSS} be satisfied. If there exist matrices $P,Q,R\succ0$ and a constant $\kappa>0$ such that
	\begin{align}\label{eq:IOSS_LMI}
	\begin{bmatrix}
	PA+A^\top P+\kappa P - C^\top RC & PB-C^\top RD\\
	B^\top P-D^\top RC & -D^\top RD-Q
	\end{bmatrix}\preceq 0
	\end{align}
	holds for all $(x,u,w)\in\mathcal{X}\times\mathcal{U}\times\mathcal{W}$, then		$U(x_1,x_2)=|x_1-x_2|^2_P$ is a smooth i-iIOSS Lyapunov function on $\mathcal{X}\times\mathcal{U}\times\mathcal{W}$ and satisfies Assumption~\ref{ass:IOSS_Lyap} with $P_1=P_2=P$ and $\lambda=e^{-\kappa}$.
\end{theorem}

{The proof is shifted to the Appendix, Section~\ref{sec:proof_thm_IOSS}. It follows similar lines as the proofs of, e.g., \cite[Th.~1]{Manchester2018} and \cite[Th.~2]{Schiller2023c}, compare also \cite{Manchester2017,Singh2017}.}

{For a fixed value of $\kappa>0$, condition \eqref{eq:IOSS_LMI} represents an infinite set of LMIs (note that \eqref{eq:IOSS_LMI} is linear in the remaining decision variables $P,Q,R$).
	These may be solved using a finite set of LMIs and standard convex analysis tools based on semidefinite programming (SDP), e.g., by applying SOS relaxations~\cite{Parrilo2003}, by embedding the nonlinear behavior in an LPV model~\cite{Sadeghzadeh2023}, or by suitably gridding the state space and verifying~\eqref{eq:IOSS_LMI} on the grid points.}

\section{Numerical example}\label{sec:example}
To illustrate our results, we consider the following system
\begin{align}
\dot{x} &= 
\begin{bmatrix}
-2k_1x_1^2 + 2k_2x_2 + w_1\\
k_1x_1^2 - k_2x_2 + w_2\\
\end{bmatrix}\\
y &= x_1 + x_2 + w_3
\end{align}
with $k_1=0.16$ and $k_2=0.0064$, which corresponds to the chemical reaction $2A\leftrightharpoons B$ from~\cite[Sec.~5]{Tenny2002}, taking place in a constant-volume batch reactor.
The additional disturbance signal $w$ is piece-wise constant over intervals of length $t_{\Delta} = 0.01$ and satisfies $w(t)\in\mathcal{W}=\{w\in\mathbb{R}^3 : |w_{i}|\leq0.1$, $i=1,2,3\}$ for all $t\geq0$, see the top left plot in Figure \ref{fig}.
We choose the initial condition $\chi = [3,1]^\top$ and the poor initial guess $\hat{\chi} = [0.1,4.5]^\top$, which constitutes a common MHE benchmark example, see, e.g.,~\cite[Sec.~5]{Tenny2002}, \cite[Example  4.38]{Rawlings2017}, \cite[Sec.~V]{Schiller2023c}.
Here, we consider the simulation length $t_{\mathrm{sim}}=5$ and additionally assume that the true trajectory $x$ satisfies {$x(t) \in\mathcal{X} = \{x\in\mathbb{R}^2 : 0.1 \leq x_i \leq 5, i=1,2\}$} for all $0\leq t\leq t_{\mathrm{sim}}$, which is reasonable under this setup.
We verify the LMI conditions~\eqref{eq:IOSS_LMI} on $\mathcal{X}\times\mathcal{W}$ in MATLAB using\footnote{
	{Here, we exploit that condition~\eqref{eq:IOSS_LMI} is linear in $x$; for a fixed value of $\lambda=0.4$ and $\kappa=-\ln\lambda$, solving~\eqref{eq:IOSS_LMI} for the vertices of $\mathcal{X}$ implies that~\eqref{eq:IOSS_LMI} holds for all $\mathcal{X}\times\mathcal{W}$ by convexity of $\mathcal{X}$.}
} YALMIP \cite{Loefberg2004} and the SDP solver MOSEK \cite{MOSEKApS2019}.
Consequently, the quadratic i-iIOSS Lyapunov function $U(x_1,x_2)=|x_1-x_2|^2_P$ with
\begin{equation*}
P =
\begin{bmatrix}
4.009 & 3.768 \\ 3.768 & 3.549
\end{bmatrix}
\end{equation*}
satisfies Assumption~\ref{ass:IOSS_Lyap} on $\mathcal{X}\times\mathcal{W}$ with $\kappa=-\ln\lambda$, $\lambda=0.4$, $Q = \mathrm{diag}(10^3, 10^3, 10^2)$, and $R = 10^2$.

We use the MHE objective~\eqref{eq:MHE_objective} with $\Gamma(\chi,x) = 2|\chi-x|_{P}^2$ and $L(w,\Delta y) = 2|w|^2_Q + |\Delta y|_R^2$ and want to perform 50 MHE updates during the simulation (in the interval $[0,t_{\mathrm{sim}}]$).
To illustrate the flexibility of the proposed MHE scheme allowing for non-equidistant sampling, we design $\mathcal{T}$ such that it contains more samples towards the beginning of the experiment, as can be seen by the blue dots in the top right plot in Figure~\ref{fig}.
This yields $\bar{\delta}=0.19$ in \eqref{eq:def_delta}. Choosing the horizon length $T=2$ satisfies~\eqref{eq:cond_T} and guarantees the convergence rate $\rho=0.86$ in~\eqref{eq:thm_2}.

We solve each MHE problem~\eqref{eq:MHE} using CasADi~\cite{Andersson2018} and IPOPT~\cite{Waechter2005}, where we employ a standard multiple shooting approach and integrate the system dynamics~\eqref{eq:sys_1} using the classic Runge-Kutta method with step size $t_{\Delta}=0.01$.
{The computations were performed on a standard laptop, which took at most $\tau_{\max} = 19.6 \, \mathrm{ms}$ per iterate of the MHE algorithm presented at the end of Section~\ref{sec:MHE_setup} for all sampling times $t_i\in\mathcal{T}$.}

\begin{figure*}[ht]
	\begin{minipage}[t]{0.48\textwidth}
		\flushright\flushbottom
		\includegraphics{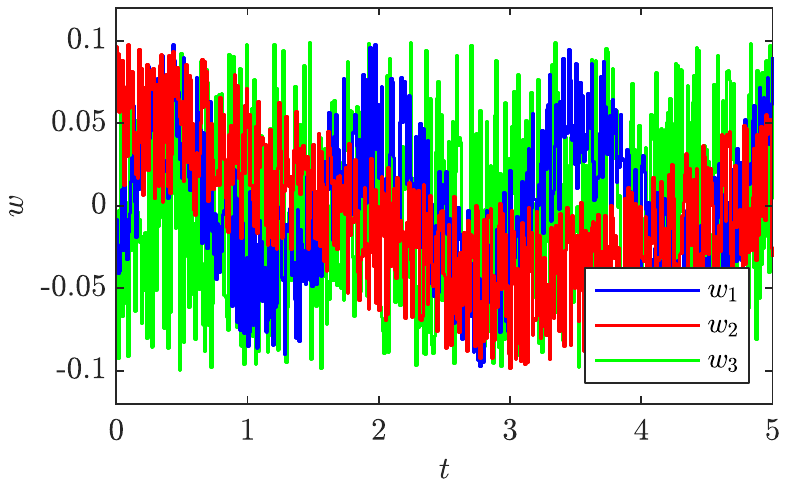}
	\end{minipage}
	\begin{minipage}[t]{0.488\textwidth}
		\flushright\flushbottom
		\includegraphics{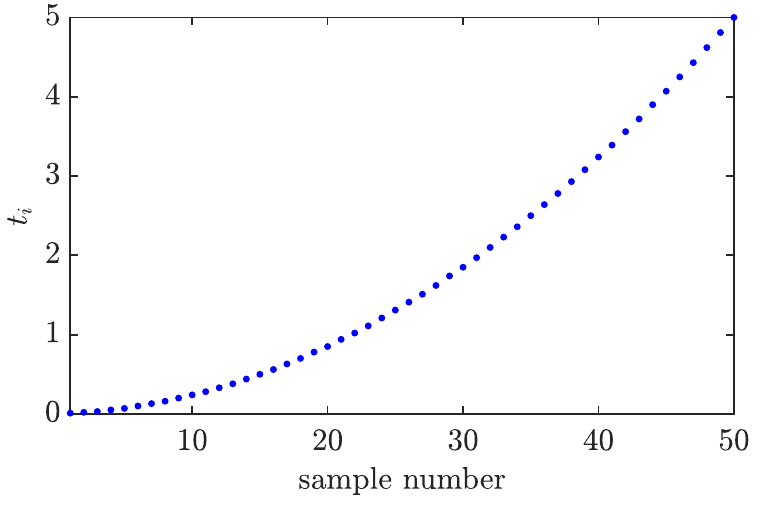}
	\end{minipage}
	
	\begin{minipage}[t]{0.485\textwidth}
		\flushright\flushbottom
		\includegraphics{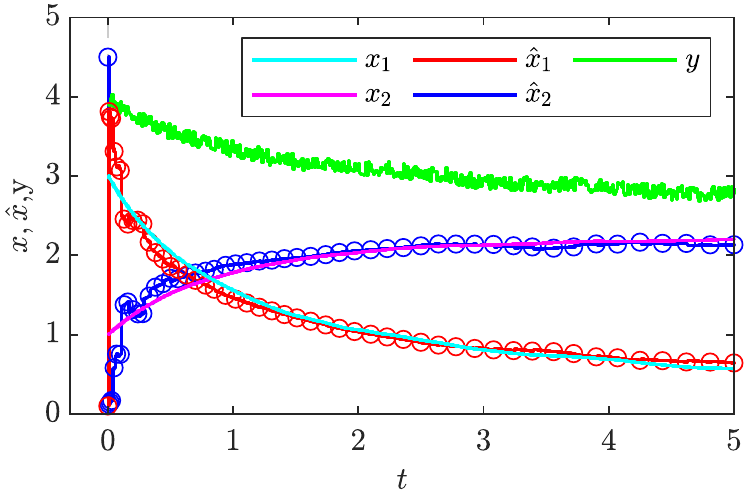}
	\end{minipage}
	\begin{minipage}[t]{0.48\textwidth}
		\flushright\flushbottom
		\includegraphics{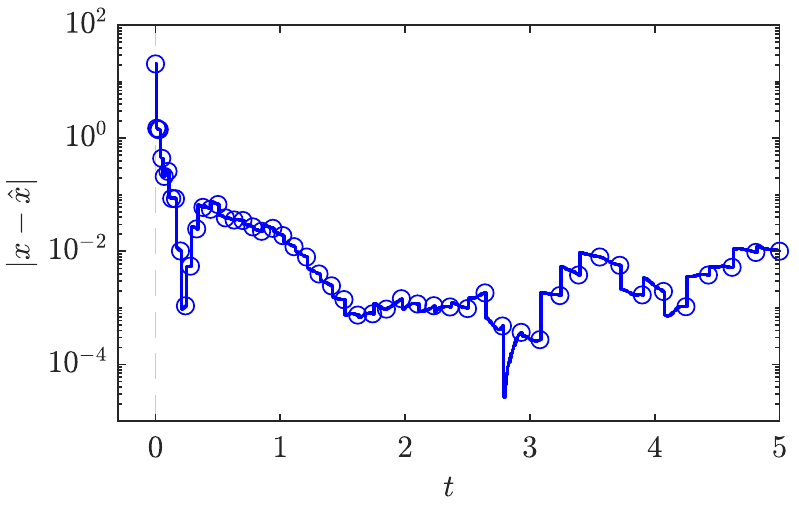}
	\end{minipage}
	\caption{Top left: disturbance signal $w$; top right: sampling times $t_i$ contained in the set $\mathcal{T}$; bottom left: comparison of the estimated trajectory $\hat{x}$~\eqref{eq:x_hat_CT}, the true system trajectory $x$, and measurements $y$; bottom right: corresponding estimation error. The circles in the bottom plots correspond to the estimates $\hat{x}(t_i)$ at the sampling times $t_i\in\mathcal{T}$.}
	\label{fig}
\end{figure*}
The estimation results and the corresponding estimation error are depicted in the bottom left and right plot in Figure~\ref{fig}, respectively.
This shows fast convergence of the error to a neighborhood around the origin, as guaranteed by Theorem~\ref{thm:MHE}.

\section{Conclusion}
In this paper, we presented an MHE scheme with discounted least squares objective for general nonlinear continuous-time systems.
Provided that the system is detectable (i-iIOSS) and admits a corresponding i-iIOSS Lyapunov function, we have shown that there exists a sufficiently long estimation horizon that guarantees robust global exponential stability of the estimation error in {a time-discounted} $L^2$-to-$L^\infty$ sense.
Moreover, we provided sufficient LMI conditions to construct i-iIOSS Lyapunov functions for certain classes of nonlinear systems, which can be efficiently verified via, e.g., SOS optimization or LPV embeddings.

While the overall robust stability analysis is based on similar ideas as in the discrete-time setting in~\cite{Schiller2023c}, the continuous-time MHE framework offers some significant advantages in practice.
{In particular, the sampling times at which the underlying optimization problem is solved can be chosen arbitrarily, which even allows the MHE scheme to be applied in an event-triggered fashion using a suitable triggering rule.
	Consequently, the proposed MHE scheme can be tailored to the problem at hand, which can yield more accurate results with less computational effort compared to standard equidistant sampling.}
Moreover, the proposed LMI conditions for computing an i-iIOSS Lyapunov function are independent of any sampling scheme and therefore less complex and easier to verify, especially when being compared to those from \cite[Th.~2, Cor.~3]{Schiller2023c} applied to the discretized system.

The applicability of the overall framework was illustrated using a nonlinear chemical reactor process.
Here, we verified the required detectability assumption by computing an i-iIOSS Lyapunov function, successfully applied the proposed continuous-time MHE scheme with non-equidistant sampling, and consequently obtained guaranteed robustly stable estimation results under practical conditions.

\subsection*{Funding}
This work was funded by the Deutsche Forschungsgemeinschaft (DFG, German Research Foundation) -- 426459964.

\appendix

\section{Proofs}

\subsection{Proof of Proposition~\ref{prop:RGES}}
\label{sec:proof_prop_RGES}

\begin{proof}
	The proof is straightforward and follows by noting that
	\begin{align*}
	&\ \int_{0}^{t_i}\eta^{t_i-\tau}\beta_{w}(|w(\tau)|)d\tau\\
	\leq&\ \int_{0}^{t_i}\eta^{t_i-\tau}d\tau\cdot \esssup\limits_{s\in[0,t_i]}\left\{\beta_{w}(|w(s)|)\right\}\\
	=&\  \frac{\eta^{t_i}-1}{\ln\eta}\beta_{w}(\|w\|_{0:t_i})\\
	\leq&\  -\frac{1}{\ln\eta}\beta_{w}(\|w\|_{0:t_i}).
	\end{align*}
	Consequently, using the fact that $a+b\leq\max\{2a,2b\}$ for $a,b\geq0$, the RGAS property from \eqref{eq:RGAS} implies that
	\begin{align}
	&\ \beta(|x(t_i,\chi,u,w)-\hat{x}(t_i)|)\nonumber\\
	\leq&\ \max\left\{2\beta_{x}(|\chi-\hat{\chi}|)\eta^{t_i}, -\frac{2}{\ln\eta}\beta_{w}(\|w\|_{0:t_i})\right\}, \label{eq:proof_RGAS}
	\end{align}
	which is equivalent to \eqref{eq:RGAS_sup} by setting $\psi(s,t) := \beta^{-1}(2\beta_x(s)\eta^{t_i})$ and $\gamma(s) := \beta^{-1}(-\frac{2}{\ln\eta}\beta_w(s))$.
	
	When the state estimator is additionally RGES, we can exploit in \eqref{eq:proof_RGAS} that $\beta(s)\geq C_1s^r$ and $\beta_x(s)\leq C_2 s^r$ with $C_1,C_2,r$ from Definition~\ref{def:RGAS}, which yields $\psi(s,t) \leq \big(2\frac{C_2}{C_1}\big)^{\frac{1}{r}}s\eta^{\frac{t_i}{r}}$ and $\gamma(s)\leq \big(-\frac{2}{C_1\ln\eta}\beta_w(s)\big)^\frac{1}{r}$ (recall that $s\mapsto s^{\frac{1}{r}}$ is strictly increasing in $s\geq0$).
	Consequently, we can choose $C:=\big(2\frac{C_2}{C_1}\big)^{\frac{1}{r}}$ and $\rho:=\eta^{\frac{1}{r}}\in[0,1)$. A suitable redefinition of $\gamma$ establishes our claim and hence concludes this proof.
\end{proof}

\subsection{Proof of Proposition~\ref{prop:MHE}}
\label{sec:proof_prop_MHE}

\begin{proof}
	Given any $t\geq0$ and its corresponding sampling time $t_i=k(t)$, let $l:=t-t_i+T_{t_i}\in[0,T_{t_i}]$ and recall that $\hat{x}(t) = \bar{x}_{t_i}^*(l) = x(l,\bar{\chi}^*_{t_i},u,\bar{w}^*_{t_i})$ by~\eqref{eq:x_hat_CT}.
	Since $\bar{x}_{t_i}^*$ satisfies~\eqref{eq:sys} on $[0,T_{t_i}]$ due to the constraints~\eqref{eq:MHE_1}-\eqref{eq:MHE_3}, we can invoke the i-iIOSS Lyapunov function from Assumption~\ref{ass:IOSS_Lyap}.
	To this end, we use that $x(t)=x(t,\chi,u,w) = x(l,x(t-T_{t_i}),u_{t_i},w_{t_i})$, where $w_{t_i}:[0,T_{t_i})\rightarrow\mathcal{M}_{\mathcal{W}}$ denotes the segment of $w$ in the interval $[t_i-T_{t_i},t_i)$ defined by $w_{t_i}(\tau) := w(t_i-T_{t_i}+\tau)$, $\tau\in[0,T_{t_i})$.
	Then, we can evaluate the dissipation inequality~\eqref{eq:IOSS_Lyap_2_ass} with the trajectories $x_1(l) = x(l,x(t_i-T_{t_i}),u_{t_i},w_{t_i})$ and $x_2(l) = \bar{x}^*_{t_i}(l) = x(l,\bar{\chi}^*_{t_i},u_{t_i},\bar{w}^*_{t_i})$ and the corresponding outputs $y_1(l) = y(l+t_i-T_{t_i}) = y_{t_i}(l)$ and $y_2(l) = \bar{y}_{t_i}^*(l)$, which yields
	\begin{align}
	&\ U(x(t),\hat{x}(t)) = U(x(t_i-T_{t_i}+l),\bar{x}_{t_i}^*(l))\nonumber\\
	\stackrel{\eqref{eq:IOSS_Lyap_2_ass}}{\leq}  &\ U(x(t_i-T_{t_i}),\bar{x}_{t_i}^*(0))\lambda^l \nonumber\\
	& + \int_{0}^{l} \lambda^{l-\tau}\big( |w_{t_i}(\tau)-\bar{w}_{t_i}^*(\tau)|_Q^2+|y_{t_i}(\tau)-\bar{y}_{t_i}^*(\tau)|_R^2\big)d\tau\nonumber\\
	\leq&\ \lambda^{-(t_i-t)}\Big(U(x(t_i-T_{t_i}),\bar{\chi}_{t_i}^*)\lambda^{T_{t_i}} \nonumber \\
	& \hspace{10ex} + \int_{0}^{T_{t_i}}\lambda^{T_{t_i}-\tau}\big( |w_{t_i}(\tau)-\bar{w}_{t_i}^*(\tau)|_Q^2\nonumber\\[-2ex]
	& \hspace{21ex} +|y_{t_i}(\tau)-\bar{y}_{t_i}^*(\tau)|_R^2\big)d\tau\Big), \label{eq:proof_th1_1}
	\end{align}
	where the last inequality follows by exploiting that $l\leq T_{t_i}$, the definition of $l$, and the fact that $\bar{x}_{t_i}^*(0)=\bar{\chi}_{t_i}^*$. Define
	\begin{align}
	\bar{U}:=&\ U(x(t_i-T_{t_i}),\bar{\chi}_{t_i}^*)\lambda^{T_{t_i}} \nonumber \\
	&\ + \int_{0}^{T_{t_i}}\lambda^{T_{t_i}-\tau}\big( |w_{t_i}(\tau)-\bar{w}_{t_i}^*(\tau)|_Q^2\nonumber\\[-1ex]
	& \hspace{15ex} +|y_{t_i}(\tau)-\bar{y}_{t_i}^*(\tau)|_R^2\big)d\tau. \label{eq:proof_th1_Uti}
	\end{align}
	By Cauchy-Schwarz and Young's inequality, we have
	\begin{equation}
	|w_{t_i}(\tau)-\bar{w}_{t_i}^*(\tau)|_Q^2 \leq 2|w_{t_i}(\tau)|_Q^2+2|\bar{w}_{t_i}^*(\tau)|_Q^2, \ \tau \in [0,T_{t_i}) \label{eq:proof_th1_2}
	\end{equation}
	and
	\begin{align}
	&\ \ U(x(t_i-T_{t_i}),\bar{\chi}_{t_i}^*)  \stackrel{\eqref{eq:IOSS_Lyap_1_ass}}{\leq} |x(t_i-T_{t_i})-\bar{\chi}_{t_i}^*|^2_{P_2} \nonumber\\
	=&\ |x(t_i-T_{t_i})-\hat{x}(t_i-T_{t_i})+\hat{x}(t_i-T_{t_i})-\bar{\chi}_{t_i}^*|^2_{P_2} \nonumber\\
	\leq &\ 2|x(t_i-T_{t_i})-\hat{x}(t_i-T_{t_i})|_{P_2}^2+2|\bar{\chi}_{t_i}^* - \hat{x}(t_i-T_{t_i})|^2_{P_2}. \label{eq:proof_th1_3}
	\end{align}
	Then, $\bar{U}$ in \eqref{eq:proof_th1_Uti} can be bounded using \eqref{eq:proof_th1_2}-\eqref{eq:proof_th1_3}, yielding 
	\begin{align}
	\bar{U} &\leq 2\lambda^{T_{t_i}}|x(t_i-T_{t_i})-\hat{x}(t_i-T_{t_i})|_{P_2}^2\nonumber\\
	&\hspace{2.5ex} +2\lambda^{T_{t_i}}|\bar{\chi}_{t_i}^* - \hat{x}(t_i-T_{t_i})|^2_{P_2} +  \hspace{-0.5ex}\int_0^{T_i}\hspace{-0.5ex}\lambda^{T_{t_i}-\tau}2|w_{t_i}(\tau)|_Q^2 d\tau \nonumber\\
	&\hspace{2.5ex}+ \int_0^{T_i}\lambda^{T_{t_i}-\tau} \big(2|\bar{w}_{t_i}^*(\tau)|_Q^2+|y_{t_i}(\tau)-\bar{y}_{t_i}^*(\tau)|_R^2\big)d\tau\nonumber\\
	&\stackrel{\eqref{eq:MHE_objective}}{=} 2\lambda^{T_{t_i}}|x(t_i-T_{t_i})-\hat{x}(t_i-T_{t_i})|_{P_2}^2\nonumber \\ 
	&\hspace*{2.5ex}+  2\int_0^{T_i}\lambda^{T_{t_i}-\tau}|w_{t_i}(\tau)|_Q^2d\tau + J(\bar{\chi}^*_{t_i},\bar{w}_{t_i}^*,\bar{y}_{t_i}^*,t_i). \label{eq:proof_th1_4}
	\end{align}
	By optimality, it follows that
	\begin{align}
	J(\bar{\chi}^*_{t_i},\bar{w}_{t_i}^*,\bar{y}_{t_i}^*,t_i)\leq&\ J(x(t_i-T_{t_i}),w_{t_i},y_{t_i},t_i)\nonumber\\
	= &\ 2\lambda^{T_{t_i}}|x(t_i-T_{t_i}) - \hat{x}(t_i-T_{t_i})|^2_{P_2}\nonumber\\
	&\  + 2\int_{0}^{T_{t_i}}\lambda^{T_{t_i}-\tau}|w_{t_i}(\tau)|_Q^2d\tau.\label{eq:proof_th1_5}
	\end{align}
	Hence, \eqref{eq:proof_th1_4} and \eqref{eq:proof_th1_5} lead to
	\begin{align}
	\bar{U} \leq&\ 4\lambda^{T_{t_i}}|x(t_i-T_{t_i})-\hat{x}(t_i-T_{t_i})|_{P_2}^2\nonumber\\
	&\ +  4\int_0^{T_i}\lambda^{T_{t_i}-\tau}|w_{t_i}(\tau)|_Q^2d\tau\nonumber\\
	=&\ 4\lambda^{T_{t_i}}|x(t_i-T_{t_i})-\hat{x}(t_i-T_{t_i})|_{P_2}^2\nonumber\\
	&\ +  4\int_{t_i-T_{t_i}}^{t_i}\lambda^{t_i-\tau}|w(\tau)|_Q^2d\tau, \label{eq:proof_th1_7}
	\end{align}
	where the last equality follows by a change of coordinates.
	In combination, we obtain
	\begin{align}
	&\ U(x(t),\hat{x}(t)) \stackrel{\eqref{eq:proof_th1_1},\eqref{eq:proof_th1_Uti}}{\leq} \lambda^{-(t_i-t)} \bar{U}\nonumber\\
	\stackrel{\eqref{eq:proof_th1_7}}{\leq}&\ \lambda^{-(t_i-t)}\Big( 4\lambda^{T_{t_i}}|x(t_i-T_{t_i})-\hat{x}(t_i-T_{t_i})|_{P_2}^2 \nonumber \\
	&\hspace{10ex} +  4\int_{t_i-T_{t_i}}^{t_i}\lambda^{t_i-\tau}|w(\tau)|_Q^2d\tau\Big).\label{eq:proof_ind_end}
	\end{align}
	Using $|x(t_i-T_{t_i})-\hat{x}(t_i-T_{t_i})|^2_{P_2} \leq\lambda_{\max}(P_2,P_1) |x(t_i-T_{t_i})-\hat{x}(t_i-T_{t_i})|_{P_1}^2$ and the first inequality in~\eqref{eq:IOSS_Lyap_1_ass} yields~\eqref{eq:thm_1}, which finishes this proof.
\end{proof}

\subsection{Proof of Theorem~\ref{thm:IOSS}}
\label{sec:proof_thm_IOSS}

\begin{proof}
	For any pair of points $x_1,x_2\in\mathcal{X}$, let $\Xi(x_1,x_2)$ denote the set of piecewise smooth curves $[0,1]\rightarrow\mathcal{X}$ connecting $x_1$ and $x_2$ such that $c\in\Xi(x_1,x_2)$ satisfies $c(0)=x_1$ and $c(1)=x_2$.
	
	Given any $\chi_1,\chi_2\in\mathcal{X}$, $u\in\mathcal{M}_{\mathcal{U}}$, and $w_1,w_2\in\mathcal{M}_{\mathcal{W}}$ satisfying Assumption~\ref{ass:IOSS}, we consider the trajectories $x_i(t) = x(t,\chi_i,u,w_i)$ and their output signals $y_i(t) = y(t,\chi_i,u,w_i)$, $t\geq0$, $i=1,2$.
	
	At any fixed time $t=t^\star\geq0$, let us consider the following smoothly parameterized paths for $s\in[0,1]$: the path of states $c(t,s)\in\Xi(x_1(t),x_2(t))$ and the paths of disturbances $\omega(t,s) = w_1(t) + s(w_2(t)-w_1(t))$.
	For $t\in[t^\star,t^\star+\epsilon)$ (with $\epsilon>0$ arbitrarily small to guarantee local existence of solutions and continuity of $u,w$ over $t\in[t^\star,t^\star+\epsilon)$) and each fixed $s\in[0,1]$, the path $c(t,s)$ evolves according to~\eqref{eq:sys} such that
	\begin{subequations}\label{eq:dIOSS_path_dyn}
		\begin{align}
		\dot{c}(t,s) &= f(c(t,s),u(t),\omega(t,s)), \label{eq:dIOSS_path_dyn_1}\\
		\zeta(t,s) &= h(c(t,s),u(t),\omega(t,s)), \label{eq:dIOSS_path_dyn_2}
		\end{align}
	\end{subequations}
	where $\zeta(t,s)$ satisfies $y_1(t)=\zeta(t,0)$ and $y_2(t)=\zeta(t,1)$ for each $t\in[t^\star,t^\star+\epsilon)$.
	Differentiating~\eqref{eq:dIOSS_path_dyn} with respect to $s\in[0,1]$ yields (after interchanging the order of differentiation of $t$ and $s$)
	\begin{subequations}\label{eq:IOSS_diff_dyn}
		\begin{align}
		\dot{\delta}_x 
		&= A\delta_x 
		+ B\delta_w, \label{eq:IOSS_diff_dyn_1}\\
		\delta_y 
		&= C\delta_x 
		+ D\delta_w \label{eq:IOSS_diff_dyn_2}
		\end{align}
	\end{subequations}
	for all $t\in[t^\star,t^\star+\epsilon)$ using the substitutions $\delta_x := dc/ds(t,s)$, $\delta_w := d\omega/ds(t,s)$, and $\delta_y := d\zeta/ds(t,s)$, and where the matrices $A,B,C,D$ are the linearizations of $f$ and $h$ as in \eqref{eq:IOSS_ABCD} evaluated at $(c(t,s),u(t),\omega(t,s))$.
	Assuming that $(c(t,s),u(t),\omega(t,s))\in\mathcal{X}\times\mathcal{U}\times\mathcal{W}$ for all $t\in[t^\star,t^\star+\epsilon)$, one can easily verify (by exploiting \eqref{eq:IOSS_diff_dyn}) that satisfaction of the point-wise LMI condition~\eqref{eq:IOSS_LMI} implies that
	\begin{equation}\label{eq:IOSS_diff}
	\frac{d}{dt}|\delta_x|^2_P \leq -\kappa |\delta_x|^2_P + |\delta_w|^2_Q + |\delta_y|^2_R
	\end{equation}
	for all $t\in[t^\star,t^\star+\epsilon)$. By integrating~\eqref{eq:IOSS_diff} over $s\in[0,1]$, interchanging integration and differentiation, and defining $E(c(t,s)):=\int_0^1|dc/ds(t,s)|^2_Pds$, we obtain
	\begin{align*}
	\dot{E}(c(t,s))
	\leq&\ -\kappa E(c(t,s)) \nonumber \\
	&\ + \int_0^1\left|\frac{d\omega}{ds}(t,s)\right|^2_Qds + \int_0^1\left|\frac{d\zeta}{ds}(t,s)\right|^2_Rds
	\end{align*}
	for $t\in[t^\star,t^\star+\epsilon)$. The integration over $t\in[t^\star,t^\star+\epsilon)$ yields
	\begin{align}
	&\ E(c(t^\star+\epsilon,s))-E(c(t^\star,s)) \nonumber \\
	\leq&\ \int_{t^\star}^{t^\star+\epsilon}-\kappa E(c(t,s)) \nonumber \\
	&\ + \int_0^1\left|\frac{d\omega}{ds}(t,s)\right|^2_Qds + \int_0^1\left|\frac{d\zeta}{ds}(t,s)\right|^2_Rds \,dt. \label{eq:IOSS_E}
	\end{align}
	Note that $E(c(t,s))$ can be interpreted as the Riemannian energy of the path $c(t,s)$. Since $P$ is constant and $\mathcal{X}$ is convex, the shortest path $\gamma(t,s)$ (with the minimum energy $E(\gamma(t,s))$ over all possible curves $c(t,s)\in\Xi(x_1(t),x_2(t))$) is, at each fixed $t\geq0$, always given by the straight line connecting $x_1(t)$ and $x_2(t)$, i.e., $\gamma(t,s) = x_1(t)+s(x_2(t)-x_1(t))$.
	Now let $c(t^\star,s) = \gamma(t^\star,s)$ in~\eqref{eq:IOSS_E} and note that $E(\gamma(t,s))\leq E(c(t,s))$ for all $t\in[t^\star,t^\star+\epsilon)$.
	Therefore, by taking $\epsilon\rightarrow0$, we have that	
	\begin{align*}
	\dot{E}(\gamma(t^\star,s))
	\leq&\ -\kappa E(\gamma(t^\star,s)) \nonumber \\
	&\ + \int_0^1\left|\frac{d\omega}{ds}(t^\star,s)\right|^2_Qds + \int_0^1\left|\frac{d\zeta}{ds}(t^\star,s)\right|^2_Rds.
	\end{align*}
	By construction of $\gamma$ and $\omega$ (in particular, the fact that their derivatives with respect to $s\in[0,1]$ are constant in $s\in[0,1]$), it follows that
	\begin{align*}
	&E(\gamma(t^\star,s))=\int_0^1\left|\frac{d\gamma}{ds}(t^\star,s)\right|^2_Pds =\left|x_1(t^\star)-x_2(t^\star)\right|^2_P,\\
	&\int_0^1\left|\frac{d\omega}{ds}(t^\star,s)\right|^2_Qds = \left|w_1(t^\star)-w_2(t^\star)\right|^2_Q,\\
	&\int_0^1\left|\frac{d\zeta}{ds}(t^\star,s)\right|^2_Rds = \left|y_1(t^\star)-y_2(t^\star)\right|^2_R,
	\end{align*}
	where the last equality follows since $h$ is affine in $(x,w)$ by Assumption~\ref{ass:IOSS} (which implies that $C$ and $D$ are constant in $s$).
	Since $t^\star\geq0$ was arbitrary, using the definition $U(x_1,x_2):=|x_1-x_2|^2_P$ we can infer that
	\begin{align*}
	\dot{U}(x_1(t),x_2(t)) \leq&\ -\kappa U(x_1(t),x_2(t)) \\
	&\ + \left|w_1(t)-w_2(t)\right|^2_Q + \left|y_1(t)-y_2(t)\right|^2_R
	\end{align*}
	for each $t\geq0$.
	Note that $U$ satisfies~\eqref{eq:IOSS_Lyap_1_ass} with $P_1=P_2=P$.
	In what follows, we will show that $U$ also satisfies~\eqref{eq:IOSS_Lyap_2_ass}.
	To this end, let $v:\mathbb{R}_{\geq0}\rightarrow\mathbb{R}_{\geq0}$ be the solution of the initial value problem
	\begin{align*}
	\dot{v}(t) &= -\kappa v(t) + \left|w_1(t)-w_2(t)\right|^2_Q + \left|y_1(t)-y_2(t)\right|^2_R\\
	v(0) &= U(\chi_1,\chi_2).
	\end{align*}
	Solving for $v$ yields
	\begin{align*}
	v(t) = e^{-\kappa t}v(0) + \int_{0}^{t}e^{-\kappa(t-\tau)}\big(&\left|w_1(t)-w_2(t)\right|^2_Q \\
	&+ \left|y_1(t)-y_2(t)\right|^2_R\big)d\tau.
	\end{align*}
	From the standard comparison theorem, we know that $U(x_1(t),x_2(t))\leq v(t)$ for all $t\geq 0$, which establishes~\eqref{eq:IOSS_Lyap_2_ass} with $\lambda=e^{-\kappa}$. 
	Consequently, $U$ is a quadratic (and hence smooth) i-iIOSS Lyapunov function that satisfies Assumption~\ref{ass:IOSS_Lyap}, which finishes this proof.
\end{proof}

\addtolength{\textheight}{-5.4cm}

\end{document}